\documentclass[12pt, draftclsnofoot, onecolumn]{IEEEtran}
\usepackage{graphicx}
\usepackage{amsthm}
\usepackage{epsfig}
\usepackage{latexsym}
\usepackage{amsfonts}
\usepackage{here}
\usepackage{bbm}
\usepackage{rawfonts}
\usepackage[latin1]{inputenc}
\usepackage[T1]{fontenc}
\usepackage{calc}
\usepackage{capitalgreekitalic}
\usepackage{url}
\usepackage{enumerate}
\usepackage{color}
\usepackage[tbtags]{amsmath}
\usepackage{amssymb}
\usepackage{upref}
\usepackage{epic,eepic}
\usepackage{times}
\usepackage{dsfont}
\usepackage{comment}
\usepackage{color}
\usepackage{cite}
\usepackage{caption}
\usepackage{epstopdf}
\usepackage{cases}

\newtheorem{theorem}{\bf Theorem}
\newtheorem{proposition}{\bf Proposition}

\newtheorem{definition}{\bf Definition}
\newtheorem{remark}{\bf Remark}

\newcounter{step}
\newlength{\totlinewidth}
  {\end{list}%
  \rule{\linewidth}{1pt}}
\newcounter{substep}

  {\end{list}}

\newlength{\aligntop}
\newlength{\alignbot}
 \makeatletter
\renewenvironment{align}{%
  \vspace{\aligntop}
  \start@align\@ne\st@rredfalse\m@ne
}{%
  \math@cr \black@\totwidth@
  \egroup
  \ifingather@
    \restorealignstate@
    \egroup
    \nonumber
    \ifnum0=`{\fi\iffalse}\fi
  \else
    $$%
  \fi
  \ignorespacesafterend%
  \vspace{\alignbot}\par\noindent
} \makeatother

\IEEEoverridecommandlockouts

\newtheorem{corollary}{Corollary}

\begin{document}
\clearpage
\title{\huge The 5G Cellular Backhaul Management Dilemma: To Cache or to Serve\vspace{-0.3cm}}
%

\author{{Kenza Hamidouche$^{1,2}$}, Walid Saad$^{2}$, M\'erouane Debbah$^{1,3}$, Ju Bin Song$^{4}$, and Choong Seon Hong$^{4}$\vspace*{0em}\\
\authorblockA{\small $^{1}$ CentraleSup\'elec, Universit\'e Paris-Saclay, Gif-sur-Yvette,France, Email: \protect\url{kenza.hamidouche@centralesupelec.fr}\\
$^{2}$Wireless@VT, Bradley Department of Electrical and Computer Engineering, Virginia Tech, Blacksburg, VA, USA, Email: \protect\url{walids@vt.edu} \\
$^{3}$Mathematical and Algorithmic Sciences Lab, Huawei France R\&D, France, Email: \protect\url{merouane.debbah@huawei.com}\\
$^{4}$Dept. of Computer Engineering, Kyung Hee University, South Korea, Emails: \protect\url{{jsong, cshong}@khu.ac.kr}\vspace{-1cm}
}\vspace*{0em}
    \thanks{This research was supported by ERC Starting Grant 305123 MORE and the  U.S. National Science Foundation under Grants CNS-1460316, CNS-1513697, and CNS-1617896.}%
  }
%
%
%
%
\maketitle
\thispagestyle{empty}
\begin{abstract}
With the introduction of caching capabilities into small cell networks (SCNs), new backaul management mechanisms need to be developed to prevent the predicted files that are downloaded by the at the small base stations (SBSs) to be cached from jeopardizing the urgent requests that need to be served via the backhaul. Moreover, these mechanisms must account for the heterogeneity of the backhaul that will be encompassing both wireless backhaul links at various frequency bands and a wired backhaul component. In this paper, the heterogeneous backhaul management problem is formulated as a minority game in which each SBS has to define the number of predicted files to download, without affecting the required transmission rate of the current requests. For the formulated game, it is shown that a unique fair proper mixed Nash equilibrium (PMNE) exists. Self-organizing reinforcement learning algorithm is proposed and proved to converge to a unique Boltzmann-Gibbs equilibrium which approximates the desired PMNE. Simulation results show that the performance of the proposed approach can be close to that of the ideal optimal algorithm while it outperforms a centralized greedy approach in terms of the amount of data that is cached without jeopardizing the quality-of-service of current requests.
\end{abstract}

\smallskip
\noindent \textbf{\emph{Keywords -}} \emph{small cell networks, Caching, heterogeneous backhaul, resource allocation, game theory, reinforcement learning.}


%
\IEEEpeerreviewmaketitle

\section{Introduction}
To cope with the continuously increasing wireless traffic and meet the stringent quality-of-service (QoS) of emerging wireless services, significant changes to modern-day cellular infrastructure are required \cite {CISCO2014}. One promising approach is to deploy small base stations (SBSs) that can provide an effective way to boost the capacity and coverage of wireless networks \cite{andrews2014will}. However, in order to benefit from this deployment of SBSs, several technical challenges must be addressed, in terms of interference management, resource allocation, and more importantly, \emph{backhaul management} \cite{andrews2014will,semiari2015matching,loumiotis2014dynamic}.

Indeed, the short-range and low-power heterogeneous SBSs must be connected to the core network through the backhaul infrastructure of the currently deployed wireless networks \cite{andrews2014will}. However, due to the dense deployment of SBSs coupled with the dramatically increasing traffic, the narrow band of the radio frequency spectrum in the range of $300$ MHz-$3$ GHz has to be shared by a large number of SBSs and used as both backhaul and access links, resulting in a congested backhaul. These capacity limitations of the backhaul links have pushed mobile network operators to exploit the available millimeter wave spectrum even though its deployment is still limited by the blockage and the atmospheric absorption. Thus, depending on the cost for the network operators and the geographical locations of the SBSs, different types of backhaul connections must coexist in 5G systems \cite{andrews2014will}. The types of backhauls that are being considered include a heterogeneous mix of wireless backhauls such as millimeter wave (mmW) and the conventional sub-6 GHz as well as wired connections via cable or fiber optical links \cite{dat2015radio,ghosh2014millimeter}. The use of such heterogeneous backhaul solutions has attracted significant attention in academia and industry recently \cite{andrews2014will,semiari2015matching} and \cite{ghosh2014millimeter}. Thus, if not properly managed, such capacity-limited and heterogeneous backhaul links can lead to significant delays when the SBSs are serving a large number of requests. One of the recently proposed solutions to cope with the backhaul bottleneck in small cell networks (SCNs) is via the use of distributed caching at the cellular network edge \cite{golrezaei2013femtocaching,poularakis2013exploiting,blaszczyszyn2014optimal,Liu2015Exploiting}. Distributed caching in SCNs is based on the premise of equipping SBSs with storage devices as well as exploiting the available storage at the user equipments (UEs) to reduce the load on the backhaul links. In particular, the SBSs can predict user requests for popular content and, then, download this content ahead of time to serve users locally, without using the backhaul.

Different caching solutions 
  for SCNs have been proposed. The authors in \cite{golrezaei2013femtocaching} propose a greedy algorithm that assigns a complete file or an encoded chunk of a file to a given SBS while minimizing the total delay. In \cite{poularakis2013exploiting}, the problem of caching coded segments at the SBSs while taking into account the random mobility of users is addressed. The work in \cite{blaszczyszyn2014optimal} proposes a geographical cache placement algorithm to maximize the probability of serving a user by the SBSs. In \cite{Liu2015Exploiting}, the authors propose a caching strategy that creates MIMO cooperation opportunities between the SBSs. In\cite{poularakis2014approximation}, a joint routing and caching problem is formulated in order to maximize the fraction of content requests served locally by the deployed SBSs. Energy efficiency of cache-enabled networks is analyzed in \cite{Perabathini2015CachingGreen}. Using tools from stochastic geometry, the authors study the conditions under which the area power consumption is minimized with respect to the base station transmit power, while ensuring a certain QoS in terms of coverage probability. Similarly, in\cite{Kumar2015Harvesting}, an online energy efficient
power control scheme is developed for a single energy harvesting SBS
equipped with caching capabilities. The authors in \cite{Zhou2015Multicast} and \cite{poularakis2014multicast} propose new caching approaches while taking into account the multicast opportunities that allow the base stations to serve part of the requests via a single multicast transmission. However, most of these works focus solely on the data being cached without taking into account the fact that such requests will be shared with other requests for data that devices require immediately rather than in the future.

Beyond caching in small cells, we note that there has been considerable works on caching in the computer science community. The idea of caching was initially introduced for central processing units and hard disk drivers and then was extended to web browsers and operating systems \cite{belady1966study}. Different approaches were considered for replacing the cached content such as as removing the least recently used or the least frequently used content \cite{belady1966study}. The closest caching models to the considered one in this paper, is caching in content delivery networks and content centric networks \cite{borst2010distributed,wang2014cache}. The idea consists in storing data at the closest proxy servers of the content delivery networks to the end users, known as the network edge. The aim from this approach is to balance the load over the servers, reduce the bandwidth requirements and thus reduce the users service time \cite{wang1999survey}. The content centric networks rely on the same idea of caching with more intelligent forwarding strategies. Indeed, the content files are identified by name instead of their location, allowing to spread the content all over the Internet network in a smart way \cite{araldo2014cost,ahlgren2012survey}. Recently, the idea of caching was introduced in cellular networks to deal with the capacity-limited backhaul in small cell networks \cite{bacstuug2014living,wang2014cache}. Despite the similarities with caching in the Internet, the network structure of SCNs is significantly different from Internet architecture. Thus, new challenges arise in SCNs such as accounting for channel characteristics and interference, that make the previously proposed approaches for the Internet not applicable, as discussed in \cite{golrezaei2013femtocaching,poularakis2013exploiting,blaszczyszyn2014optimal,Liu2015Exploiting}. This led to the recent emergence of a large literature that aims to address the caching problem while taking into account the specific characteristic of SCNs, as discussed previously. 
  
Moreover, several works \cite{semiari2015matching,loumiotis2014dynamic}, and \cite{liebl2012fair,yi2012backhaul,sengupta2009economic} have addressed the backhaul management problem in order to satisfy the required transmission rate of the SBSs. The main challenge is determine the backhaul resource blocks that should be allocated to each demanding SBS allowing the SBSs to satisfy the QoS requirement of their served users. The authors in \cite{semiari2015matching}, propose a backhaul allocation approach using matching theory in order to allocate the required data rate to each SBS while considering mmW backhaul capabilities. In \cite{loumiotis2014dynamic}, an evolutionary game model for dynamic backhaul resource allocation is proposed while taking into account the dynamics of users' traffic. The  authors in \cite{liebl2012fair} propose a fair resource allocation model for the out-band relay backhaul links, enabled with channel aggregation. The aim of this approach is to maximize the throughput fairness among backhaul and access links in LTE-Advanced relay system. In \cite{yi2012backhaul}, a backhaul resource allocation approach is proposed for LTE-Advanced in-band relaying. This  approach optimizes resource partitioning between  relays  and  macro  users,  taking  into  account  both backhaul and access links quality. In \cite{sengupta2009economic}, an economic model is proposed to allow spectrum providers to lease the backhaul resources to different operators dynamically, by using novel pricing mechanisms.

Despite being interesting, the SCN caching strategies proposed in existing works \cite{golrezaei2013femtocaching,poularakis2013exploiting,blaszczyszyn2014optimal,Liu2015Exploiting,poularakis2014approximation,Perabathini2015CachingGreen,Kumar2015Harvesting,Zhou2015Multicast,poularakis2014multicast} \emph{do not consider the impact of downloading predicted files on the other urgent non-predicted files nor do they account for the heterogeneity of the cellular backhaul links}. In fact, in a cache-enabled system, when an SBS receives a request, if it could not predict it in advance and the requested file is not available in its cache, then the request is considered as being urgent and it must be served instantaneously from the backhaul. Meanwhile, the SBS has to also download predicted files in order to be cached for serving locally the upcoming predicted requests. However, due to the limited capacity of the heterogeneous backhaul links, downloading the predicted files in order to be cached can affect the QoS experienced by the users that are served directly through the backhaul. This results from the fact that the capacity of the radio links that connect the SBSs to the UEs is usually higher than the backhaul capacity due to the hyper-dense nature of SCNs \cite{andrews2014will}. On the other hand, existing backhaul allocation approaches such as in \cite{semiari2015matching} and \cite{loumiotis2014dynamic} also do not account for the impact of caching on the required backhaul rate by each SBS. Such existing approaches may allocate backhaul for downloading predicted files whereas shifting the download of these files to off-peak hours can ensure the required transmission rate for serving the current requests. Moreover, due to the uncertainty in the prediction of the requests, the predicted files may or may not be requested by the users in the future, which makes them less critical than actual demands. The impact of this criticality factor on the backhaul usage and users' current requests has indeed been ignored in the existing literature \cite{golrezaei2013femtocaching,poularakis2013exploiting,blaszczyszyn2014optimal,Liu2015Exploiting,poularakis2014approximation,Perabathini2015CachingGreen,Kumar2015Harvesting,Zhou2015Multicast,poularakis2014multicast}.
The differentiation of request types is important for practical scenarios in which the caches should be refreshed over short time periods due to the high popularity fluctuation of the most popular files that is in the order of hours. Moreover, it allows the SBSs to deal with traffic load in offline caching models in which new peaks of traffic might emerge when all the SBSs refresh their caches simultaneously. In addition to the traffic variation, the SBSs have limited computing and communication resources which make it difficult for them to process large amounts of data and thus the caches must be refreshed more frequently. When such online caching policies are used at the SBSs, new backhaul management frameworks should be deployed to define when the predicted files should be download by the SBSs so that this additional traffic does not jeopardize the QoS of the users requesting files that are not cached at the SBSs and need to be served instantaneously.
  
The main contribution of this paper is to propose a novel distributed backhaul management approach in a wireless cellular network having caching capabilities and a heterogeneous backhaul. In particular, we propose a novel framework using which the SBSs can determine the number of predicted files to download at each time stage, without affecting the download rate of the current critical requests. We consider a SCN with different coexisting backhauls including wired links, mmW and sub-6 GHz bands that can only support a limited number of files at each time period. The problem is then formulated as a minority game (MG), in which the SBSs are the players that must decide independently, on the number of predicted files to download while taking into account other SBSs' decisions. We study the properties of the game and prove that there exists a unique fair proper mixed Nash equilibrium (PMNE) in which all SBSs have an equal chance of using the backhaul. Moreover, we propose a self-organizing reinforcement learning (RL) algorithm with incomplete information that allows the SBSs to reach a Boltzmann-Gibbs equilibrium without communicating with one another. Also, we provide a formal proof of the convergence of the RL algorithm to a unique Boltzmann-Gibbs equilibrium which approaches the PMNE in the formulated game. The proposed approach allows the SBSs to take their decisions autonomously and manage the optimization operations locally without coordinating with one another or with a centralized entity. In fact, having such self-organizing SBSs is of high importance in 5G systems due to the high density of SBSs and the capacity-limited backhaul links \cite{andrews2014will,semiari2015matching,loumiotis2014dynamic}. To our knowledge, this is the \emph{first work that jointly considers SCN backhaul management with caching by taking into account the impact of having predicted and current user requests on the backhaul allocation in cache-enabled SCNs}. Simulation results show that the amount of cached data can be 50\% higher compared to the centralized algorithm due to the reduction of information exchange. Moreover, the performance of the proposed algorithm will match the performance of the optimal, ideal centralized algorithm in more than 85\% of the cases, under properly chosen parameters.

The rest of this paper is organized as follows. Section \ref{model} presents the system model. In section \ref{form}, we formulate the problem as an MG and study its properties. In Section \ref{algo}, a distributed RL algorithm is proposed and its convergence to a unique Boltzmann-Gibbs equilibrium is proved. Section \ref{sim} provides the simulation results and Section \ref{conclusion} concludes the paper.
\section{System Model}
\label{model}
Consider a small cell network composed of a set $\mathcal{M}$ of $M$ micro base stations (MBSs) and a set $\mathcal{N}$ of $N$ SBSs. Each SBS can be connected to the MBSs via one or many backhaul links of different types which can be cable, mmW band or sub-6 GHz band. Such heterogeneous backhauls have been proposed recently as a solution to improve SCN performance as discussed in \cite{andrews2014will}. An illustration of the system model is given in Fig.~\ref{figure:sys}. The wireless link is divided into two sets of backhaul resource blocks denoted by $\mathcal{K}_1$ and $\mathcal{K}_2$ for mmW band and sub-6 GHz band, respectively. Then, depending on the required rate by each SBS, the backhaul resource blocks are allocated to the SBSs. The wired link of maximum capacity $C_{\text{max}}$ is assumed to be shared by many SBSs.
\begin{figure}
\centering
\includegraphics[scale=0.5]{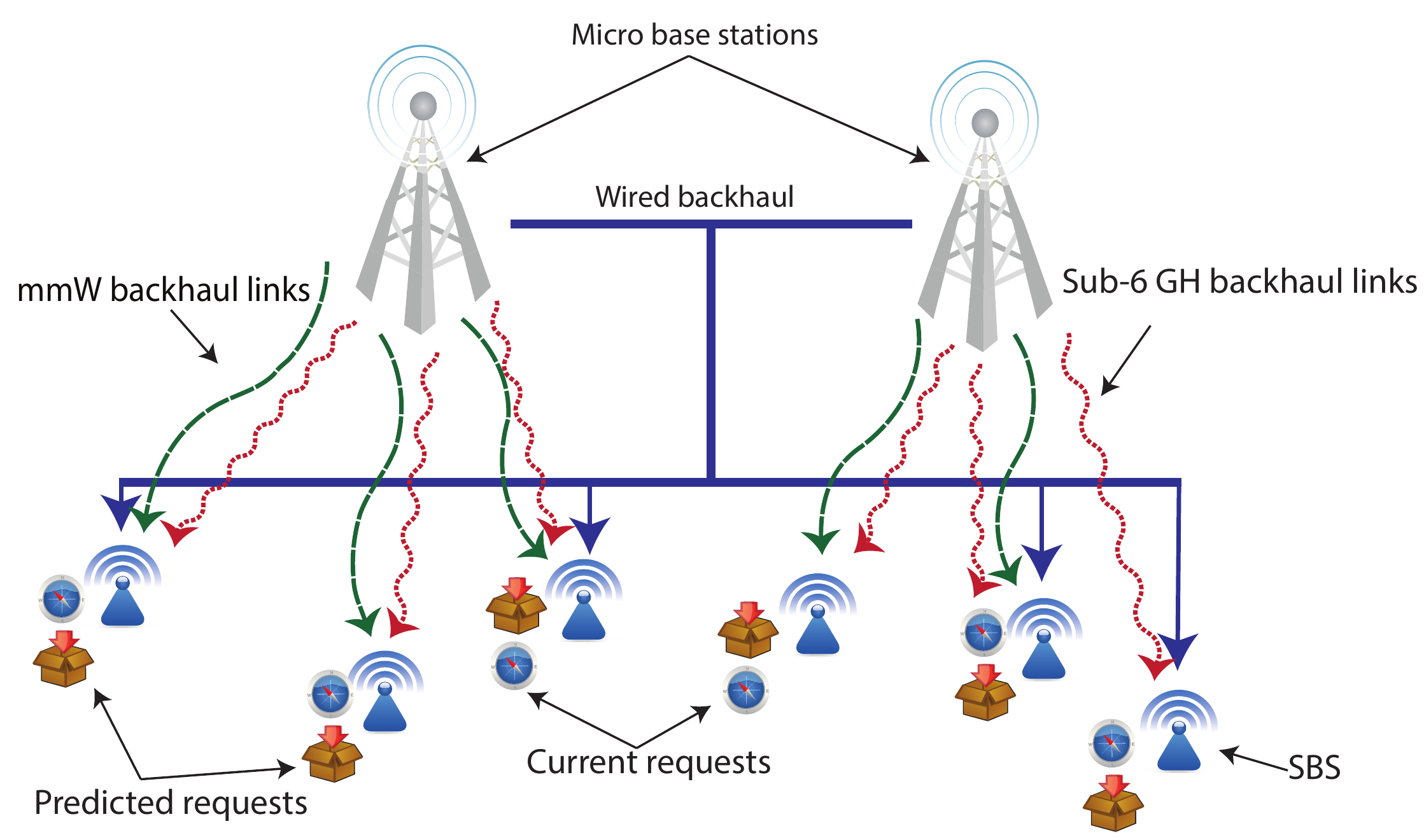}
\caption{System model.\vspace{-0.6cm}}
\label{figure:sys}
\end{figure}
The maximum achievable backhaul rate for a given SBS over the wireless backhaul, is subject to different effects such as interference between the transmitting MBSs when using sub-6 GHz band and atmospheric attenuations when using the mmW band. Indeed, since mmW bands operate at high frequencies, an antenna at a given MBS is able to provide high directional gain and thus the signals do not interfere with one another. However, the transmission rate over the mmW band is limited by rain and atmospheric attenuations as well as the distance between the transmitting MBS and the receiving SBS. For mmW, the path loss is given by\cite{ghosh2014millimeter}:
\begin{equation*}
L_{mn}^{\text{mmW}}=\beta+\alpha 10 \text{log}_{10}(\delta_{mn})+\mathcal{X},
\end{equation*}
where $\alpha$ is the slope of the fit, $\delta$ is the distance between the MBS and the served SBS, $\beta$ is the path loss for 1 meter of distance, and $\mathcal{X}$ is the deviation in fitting which is a Gaussian random variable with zero mean and variance $\zeta^2$. The signal-to-noise ratio (SNR) at a receiving SBS $n$ in the mmW band is given by:
\begin{equation}
\gamma_{mk_1n}=\frac{10 \text{log}_{10} (P_{mk_1n})-L_{mn}^{\text{mmW}}}{N_1},
\end{equation}
where $P_{m1}$ denotes the transmission power of the MBS $m$ serving SBS $n$ over backhaul resource block $k_1\in \mathcal{K}_1$ and $N_1$ is the variance of the receiver's Gaussian noise.
For sub-6 GHz bands, the rate of an SBS is usually limited by the interference experienced from the other transmitting MBSs. The signal-to-interference-plus-noise (SINR) at a receiving SBS $n$ in the sub-6 GHz band is given by:
\begin{equation}
\gamma_{mk_2n}=\frac{P_{mk_2n}|h_{mk_2n}|^2}{N_2+\sum_{i\in\mathcal{M},i\neq m}{{P_{ik_2n}|h_{ik_2n}|^2}}},
\end{equation}
where $P_{m2}$ denotes the transmission power of the MBS $m$ serving SBS $n$ over backhaul resource block $k_2\in \mathcal{K}_2$. In addition, $h_{mk_2j}$ and $N_2$ represent, respectively, the channel state of the link between MBS $m$ and SBS $j$ over backhaul resource blocks $k_2$ and the variance of the receiver's Gaussian noise. For the wired backhaul, even though the transmission is interference-free, the achievable capacity by a given SBS is limited by the number of SBSs that are served using the same link since all the served SBSs share the wired capacity $C_{\text{max}}$.

An SBS is assumed to have \emph{current} requests that is not willing to cache and \emph{predicted} requests that it has to download and cache to serve these popular requests locally without using the backhaul links. The SBSs define the set of predicted files using an underling caching policy that accounts for the available storage space, files popularity and different systems parameters such as users' mobility and SBSs' geographical positions. The set of current requests is composed of all the files that do not belong to the predefined set of predicted files and for which the SBSs receive requests. Hence, downloading the files to serve the predicted requests during high traffic times will affect badly the backhaul rate of the SBS for serving the current requests due to congestion in the wired backhaul or interference in the wireless backhaul. Assume that, at a given time period, an SBS needs a rate $R_{n}$ to serve all the current requests and rate $D_{n}(s_n)$ to download $s_n$ files to serve the predicted requests. In order to serve the requests, a backhaul allocation algorithm is used to assign each backhaul resource block to a given SBS. Without loss of generality, we assume that an algorithm such as the one proposed in \cite{semiari2015matching} for mmW and sub-6 GHz backhaul resource blocks is used in this context. The algorithm results in an assignment of SBSs to the backhaul resource blocks that aims to satisfy the required rate by each SBS. However, the requested rate by each SBS depends on the number of files that each SBS requests. Thus, the output of the backhaul allocation algorithms is a function of the global rate $R=[R_1,...,R_N]$ that is required for serving the current requests, and the requests profile of predicted files requested by the SBSs denoted $\mathcal{F}_c=[\mathcal{F}_{c,1},...,\mathcal{F}_{c,N}]$ of cardinalities $[s_1,...,s_N]$, and is given by a matrix $\boldsymbol{\eta}_k(\mathcal{F}_c,R) \in \{0,1\}^{M\times N}$, for each backhaul resource block $k\in\mathcal{K}\triangleq\mathcal{K}_1\cup\mathcal{K}_2$. An entry $\eta_{mkn}(\mathcal{F}_c,R)$ of the matrix $\boldsymbol{\eta}_{k}(\mathcal{F}_c,R)$ equals 1 if MBS $m$ allocates backhaul resource block $k$ to SBS $n$, and equals 0 otherwise. We use $f_c=\sum_{n\in\mathcal{N}}{s_n}$ to denote the cardinality of the set $\mathcal{F}_c$, which corresponds to the total number of predicted files that all the SBSs decide to download. Given the backhaul resource blocks assignment algorithm, the total achievable backhaul rate for SBS $n$ is given by:
\begin{equation}
r_{n}=\sum_{m\in\mathcal{M}}{\big [c_{mn}(\mathcal{F}_c,R)+\sum_{k\in \mathcal{K}}{\omega_k \text{log}(1+\gamma_{mkn}(\eta_{mkn}(\mathcal{F}_c,R))) }\big]},
\end{equation}
where $\omega_k$ is the bandwidth capacity of backhaul resource block $k$. Since an SBS perceives only the interference from the MBSs transmitting over the same resource blocks, we rewrite the interference as a function of the outcome of the backhaul assignment algorithm $\eta_{mkn}(\mathcal{F}_c,R)$. $c_{mn}$ is the wired allocated capacity by MBS $m$ to SBS $n$. The wired backhaul link's capacity is assumed to be shared between all the SBSs based on the remaining traffic load that could no be served through the wireless backhaul. The allocated wired backhaul by MBS $m$ to SBS $n$ is given by, $c_{mn}=\sigma_{n}(\mathcal{F}_c,R)c^{\prime}_{m}$, where $c^{\prime}_{m}$ is the available wired backhaul capacity at MBS $m$ and $\sigma_{n}(\mathcal{F}_c,R)=\frac{\sum_{f\in\mathcal{F}_{c,n}}q_{f}+R_n}{\sum_{n\in\mathcal{N}}\Big[\sum_{f\in\mathcal{F}_{c,n}}q_{f}+R_n\Big]}$ is the traffic load of SBS $n$ over the total traffic load of all the SBSs, where  $\mathcal{F}_{c,n}$ the set of predicted files that is requested by SBS $n$ and $q_{f}$ the maximum required data rate by an SBS $n$ to serve each file $f\in\mathcal{F}_{c,n}$ which could depend on the type of the application, and the paid price by the users requesting that file for the service. Hence, the backhaul capacity $c_{mn}(\mathcal{F}_c,R)=\sigma_n(\mathcal{F}_c,R)c^{\prime}_m$ that is assigned to each SBS $n$ is proportional to the traffic load of SBS $n$ as compared to the traffic load of other SBSs.

Based on the total capacity of the heterogeneous backhaul and the number of urgent requests, each SBS has to decide, without a direct communication with the other SBSs, on the number of predicted files to download without reducing the transmission rate of the current requests in the network. This problem is formulated in the next section, as a minority game.

\section{Problem Formulation}
\label{form}
The considered problem is characterized by two main properties which are the limited capacity of the backhaul links and the possibility of delaying the predicted requests in cache-enabled small cell networks. Moreover, the achievable reward by the SBSs in such networks is dependent on the chosen actions by all other SBSs. The most suitable tool to can account for all these properties is the class of minority games. In minority games, the SBSs are enforced to cooperate without any coordination between the SBSs.  Such algorithms are important for 5G networks in which the SBSs will be deployed ultra-densely and the capacity-limited backhaul links make it difficult to support any additional coordination load.

\subsection{Backhaul Management Minority Game (BMMG)}
We formulate the problem of backhaul management as a one stage MG, in which the SBSs are the players and each of them has to determine the number of predicted files that must be downloaded from the core network at a given time period, without coordinating with the other SBSs. We consider that, regardless of the traffic load, the SBSs must serve the urgent requests whenever they receive them but they have to decide whether to download or not files that can be cached to serve the predicted requests. Depending on the traffic load, assume that the maximum number of predicted files that can be downloaded at a given time period without affecting the service of the current requests is given by $\phi \in[0,F]$, where $F$ is the cardinality of the set of files  $\mathcal{F}$ from which users can pick their requests. It should be noted that the value of $\phi$ is fixed for the considered time period but can vary from a given time stage to another one. Moreover, the statute of a given request can evolve over time, from a predicted request to an urgent request. Since we consider a one stage MG, there is no need to account for the evolution of the requests as the statute's changes are implicitly considered by defining the static sets of current and predicted requests at each time period. However, we do not make any restriction on the cooperation between the SBSs. Thus, when determining the caching policy, a given file can be divided into small chunks each of which will be cached at a different SBS. Each file chunk will be considered by a SBS as a complete files in our model and will add it to the set of predicted files. In this model, the storage space is allocated more efficiently and a user can be served by multiple SBSs at the same time.

In an MG, each SBS $n$ has to select a strategy $s_n$ from a set $S_n=\{0,1,..,F_n\}$, where $F_n$ corresponds to the number of files for which the SBS $n$ predicts requests and these files must be cached at the SBS. Note that even by caching the predicted files, these files may not be requested by the users in the future which can result in a waste of backhaul capacity if the critical urgent requests are not prioritized.  Moreover, the files in the set of predicted files can become current requests if the SBSs are not able to cache these files before the users request them. In this case, the files are removed from the set of predicted requests and added to set of current requests to serve them instantaneously. The capacity $\phi$ represents the limit starting from which the utility of the players will begin to decrease. Indeed, assuming that all the files have the same size, if the SBSs decide to download a large number of predicted files, this will reduce the allocated backhaul rate per SBS and hence degrade the QoS of the requests that are currently being served from the backhaul, as these urgent files will not be served on time. This is equivalent to deciding on the number of backhaul resource blocks that an SBS needs to use at each time period, as the higher is the number of files an SBS decides to download, the higher is the number of backhaul resource blocks that must be assigned to that SBS. Thus, an SBS delays the service of its own current requests if the total number of predicted files that are requested by the SBSs exceeds $\phi$.

The formulated game is classified as a minority game \cite{challet2013minority}, due the limited number of predicted files that can be supported by the backhaul links, as well as the nature of the SBSs' utility. Essentially, in an MG, players are always better off when they select the action selected by the minority group. The size of the minority group is determined by the maximum system resources that can be allocated to the players. In our context, an SBS would prefer not to request predicted files if more than $\phi$ predicted files are requested by the SBSs, in which case the set of SBSs not requesting files will constitute the minority group. Similarly, the SBSs would prefer to request predicted files if less than $\phi$ files are requested by the other SBSs. The minority group in this case corresponds to the SBSs that choose to request predicted files. The main challenge in this game is that the SBSs do not communicate with one another and if they all think that the backhaul will be congested, none of the SBSs will requests files and the backaul will be underused. On the other hand, if all the SBSs think that the other SBSs will not request predicted files, the backhaul will be congested and the utility of the SBSs decreases.

The utility of an SBS $n$ when it decides to download $s_n$ predicted files, is given by:
\small
\begin{equation}
  \left\{
\begin{aligned}
u_n(s_n,\mathcal{F}_c)=  -R_{n}-D_n(s_n) +\sum_{m\in\mathcal{M}}{\Big (c_{mn}(\mathcal{F}_c,R)} 
+\sum_{k\in \mathcal{K}}{\omega_k \text{log}(1+\gamma_{mkn,t}(\eta_{mkn}(\mathcal{F}_c,R))) }\Big), \text{ if }f_c\geq\phi, \\ 
u_n(s_n,\mathcal{F}_c) = R_{n}+D_n(s_n) -\sum_{m\in\mathcal{M}}{\Big (c_{mn}(\mathcal{F}_c,R)} 
+\sum_{k\in \mathcal{K}}{\omega_k \text{log}(1+\gamma_{mkn,t}( \eta_{mkn}(\mathcal{F}_c,R)))}\Big),\text{ if } f_c\leq\phi,
\label{system}
\end{aligned}
  \right.
\end{equation}
\normalsize
where $f_c$ is the total number of requested predicted files by all the SBSs. This utility represents the difference between the allocated backhaul rate for SBS $n$ and the rate it requires to serve all the current requests and the $s_n$ predicted requests. The required rate for serving the current requests and predicted requests can be given by $R_n=\sum_{f\in\mathcal{F}_n^{\prime}}\frac{L_{f}}{x_f}$ and $D_n(s_n)=\sum_{f=1}^{s_n}\frac{L_{f}}{x_f}$, respectively, where $x_{f}$ is the minimum time during which the request of the $f$th file should be served and $\mathcal{F}^{\prime}_n$ is the set of current files of SBS $n$. We point out here that we do not consider the specific files as we assumed that every SBS defines a priority order for requesting the files from the MBSs. Thus, the set of the first $s_n$ files is unique. The required rate is the fraction between the required time for serving the file and the size of the file $L_f$.

Note that when the maximum backhaul capacity is reached, i.e. $f_c\geq\phi$, the higher is the number of requested files by the SBSs, the lower is the number of assigned backhaul resource blocks and wired capacity to the SBSs. Thus, the utility of a given SBS is a decreasing function of the total number of the requested files by the SBSs. Moreover, in order to avoid underusing the backhaul, the utility obtained by an SBS that chooses not to request predicted files when $f_c \leq \phi$, is defined as an increasing function of the number of requested files until all the backhaul is efficiently allocated, i.e., $f_c=\phi$. The exact value of $\phi$ can be determined based on the backhaul assignment algorithm provided in [3] based on the maximum number of files that are supported by all the SBSs. In fact, since the SBSs always start by caching the files for which the SBSs expect receiving the requests sooner as compared to other requests, then the value of $\phi$ is unique. As shown in (3), it should be noted that the data rate that is achieved by every SBS depends not only on the requested rate by all the SBSs for serving the predicted requests but also on the set of current files that is requested by the SBSs. Thus, the utility of each SBS also depends on the strategies selected by all other SBSs.

Having defined the utility functions,  the goal is to find a solution for the defined game. For this, we distinguish between the pure strategy and proper mixed strategy cases.
\subsubsection{Pure Strategies}
In the pure strategy game, each SBS selects its strategies deterministically, i.e., with probability $1$ or $0$. The \emph{pure Nash equilibrium} is defined as follows \cite{myerson1978refinements}.

\begin{definition}
Let $s_n$  be the strategy selected by SBS $n\in\mathcal{N}$ and $\boldsymbol{s}_{-n}=[s_1,...,s_{n-1},s_{n+1},...,s_{N}]$ the strategy profile of all the other SBSs except SBS $n$. A strategy profile $\boldsymbol{s}^*=[s^*_1,..,s^*_{N}]$ is a \emph{pure Nash equilibrium} (PNE) if:
\begin{equation}
\forall n, s_n\in S_n, u_n(s^*_n,s_{-n}^*)\geq u_n(s_n,s^*_{-n}).
\end{equation}
\end{definition}
In MG literature, results on the existence of PNEs were provided when the number of strategies is the same for all the players and equal to two \cite{challet2013minority,moro2004minority}. However, in the formulated BMMG, each SBS has a larger set of strategies which changes from an SBS to another SBS. For the BMMG, we can derive the following result:
\begin{theorem}
There exists a PNE obtained when the total number of predicted files that are requested by the SBSs at the considered time stage, equals $\phi$.
\end{theorem}
\begin{proof}
A PNE is the state in which none of the SBSs can improve its utility by unilaterally changing its strategy. Denoting $s_n^*$ the strategy chosen by SBS $n$ in the PNE. When an SBS changes its strategy from $s^*_n$ to $s_n$, two cases can be considered: $s_n>s_n^*$ and $s_n<s_n^*$. Thus, at the PNE, the two following conditions must be satisfied:
\begin{numcases}{}
u_n(s_n^*,\phi)\geq u_n(s_n,\phi+(s_n-s^*_n)) \text{ if } s_n>s_n^*,\\
u_n(s^*_n,\phi)\geq u_n(s_n,\phi-(s^*_n-s_n)) \text{ if } s_n<s_n^*.
\end{numcases}
From (4), we can deduce that if the SBS selects another strategy $s_n>s_n^*$, then $u_n(s^*_n,\phi)\geq u_n(s_n,\phi+(s_n-s^*_n))$. This is because the utility is a decreasing function of the total number of requested files when $f_c\geq \phi$, which is the case when $\phi+(s_n-s^*_n)>\phi$. \\
On the other hand, assuming $\phi>0$ and $s_n<s_n^*$, we have $u_n(s_n,\phi)\geq u_n(s_n,\phi-(s^*_n-s_n))$. This is due to the fact that the utility is increasing by increasing the number of requested files when the total number of requested files does not exceed $\phi$, which occurs when SBS $n$ chooses a strategy $s_n<s_n^*$.\\
From these two cases, we can conclude that $f_c=\phi$ is a PNE.
\end{proof}
 In the pure strategy case, we can notice that any combination of strategies that satisfy $\sum_{n=1}^{N}s_n= \phi$ is a PNE, resulting in a large number of equilibria. Thus, in the repeated BMMG it is difficult to capture the frequency with which each SBS downloads predicted files over a large time horizon. In fact, for a given available backhaul capacity, a subset of SBSs may keep requesting a large number of files with probability 1 at each time period, \emph{while other SBSs never download any predicted files}. In order to ensure fairness between the SBSs, in terms of backhaul usage over a large time duration, we consider the proper-mixed strategy case in which each SBS $n$ selects one of the strategies $s_i\in S_n$ with a given probability $p^{(n)}_i \in(0,1)$, thus allowing a fairer backhaul use as shown next.
\subsubsection{Proper-Mixed Strategies}
\vspace{0.2cm}
In the mixed strategy game, an SBS $n \in\mathcal{N}$ can play the strategies in $ S_n$ with a probability profile $\boldsymbol{p}^{(n)}=[p^{(n)}_1,...,p^{(n)}_{F_n}]$, where $p^{(n)}_i\in(0,1)$ \cite{myerson1978refinements}. The expected utility for an SBS $n$ when choosing each of the strategies $c$ and $d$ are given respectively, by:
\begin{align}
\bar{u}_n(c, \boldsymbol{p}_{-n})=&\prod_{i\neq n}^{G}(1-p_i) u_n(c,1)+\sum_{i\neq n}^{G}p_i\prod_{j\neq \{i,n\}}^{G}(1-p_j) u_n(c,2)
+...+\prod_{l\neq n}^{N}p_lu_n(c,G).\\
\bar{u}_n(d,\boldsymbol{p}_{-n})=& \prod_{i\neq n}^{G}p_i u_n(d,1)+\sum_{i\neq n}^{G}(1-p_i)\prod_{j\neq \{i,n\}}^{G}p_j u_n(d,2)
+...+\prod_{l\neq n}^{G}(1-p_l)u_n(d,G),
\end{align}
\vspace{0.2cm}
where $G$ is the cardinality of the set $\mathcal{G}$, $p_i$ is the probability that SBS $i$ downloads its assigned file, and $\boldsymbol{p}_{-n}$ is the probability profile of all SBSs except SBS $n$. The desirable  solution concept in such systems is the proper mixed Nash equilibrium that can be defined as follows.
\begin{definition}
A \emph{proper mixed Nash equilibrium} (PMNE) specifies an optimal mixed strategy $\boldsymbol{p}^{(n)*}$ for each SBS $n\in \mathcal{N}$ such that:
\begin{equation}
\begin{aligned}
\bar{u}_n(\boldsymbol{p}^{(1)*},...\boldsymbol{p}^{(n-1)*},\boldsymbol{p}^{(n)*},\boldsymbol{p}^{(n+1)*},...,\boldsymbol{p}^{(N)*})\geq 
\bar{u}_n(\boldsymbol{p}^{(1)*},...,\boldsymbol{p}^{(n-1)*},\boldsymbol{p}^{(n)},\boldsymbol{p}^{(n+1)*},...,\boldsymbol{p}^{(N)*}),
\end{aligned}
\end{equation}
where $\bar{u}_n$ is the expected utility of SBS $n$ when the used probability profile by all the SBSs is $\boldsymbol{p}=[\boldsymbol{p}^{(1)},...,\boldsymbol{p}^{(N)}]$.
\end{definition}
 Even though proving the existence and uniqueness of the PMNE is possible when the number of players is two and the number of strategies is also equal to two, it is very challenging to extend this results to the case of multiple players and strategies even for MGs \cite{challet2013minority,moro2004minority}. The main challenge in finding the mixed strategies equilibrium is in the computation of the equilbria due to complexity of the system of equations that should be solved to find the different probabilities per player.
 In order to solve the BMMG, we reduce the problem to a minority game with multiple players, each of which has two strategies. We study this simplified game and then map the results to the original backhaul management MG.
\vspace{-0.1cm}
\subsection{Simplified Backhaul Management Minority Game (SBMMG)}
\vspace{-0.1cm}
To cast the backhaul management problem as a simplified minority game, we introduce an additional set $\mathcal{V}=\bigcup_{n\in\mathcal{N}}\mathcal{V}_n$ of $V$ virtual SBSs. In fact, for each real SBS $n$ that has a strategy set composed of $F_n$ predicted files, we consider one real SBS $n$ and create a set $\mathcal{V}_n$ of $F_n-1$ virtual SBSs. In this modified model, each real and virtual SBS $n \in \mathcal{G}=\mathcal{N}\cup\mathcal{V}$ is assigned one predicted file, i.e., $s_n=1$, and has to decide whether to download or not that predicted file. The strategy set for all the SBSs becomes a binary set $S=\{c,d\}$, in which the strategy $c$ corresponds to requesting and caching the file, and $d$ corresponds to not requesting the file from the MBSs. The set of current requests for the real SBSs is the same while it is empty for the virtual SBSs and consequently $R_{n}=0,~$ $\forall n \in\mathcal{V}$. The utilities for an SBS $n$ of choosing a strategy $c$ or $d$ when the total number of predicted files that will be requested by all virtual and real SBSs equals $f_c$, are given by:
\begin{equation}
\left\{
\begin{array}{l}
  u_n(c,f_c)=  -R_{n}-D_n(s_n)   +\sum_{m\in\mathcal{M}}{\Big (c_{mn}(\mathcal{F}_c,R)+\sum_{k\in \mathcal{K}}{\omega_k \text{log}\Big(1+\gamma_{mkn}(\eta_{mkn}(\mathcal{F}_c, R))\Big) }\Big)},\\
u_n(d,f_d)= -u_n(c,f_c+1),
\end{array}
\right.
\end{equation}
where $f_d=\sum_{n\in \mathcal{N}} F_n-f_c$ is the number of predicted files that the SBSs decide to not download.

In this SBMMG formulation, the number of requested files corresponds to the number of SBSs using strategy $c$. In an MG, a player is better off if it chooses the strategy chosen by the minority. In our context, the SBS gets a positive utility, i.e. serves the predicted file without affecting the QoS of the current requests, if the total number of predicted files that is requested by all the SBSs does not exceed $\phi$. On the other hand, if the SBS chooses not to download the predicted file while the total number of requested files does not exceed $\phi$, the SBS gets a negative utility. This represents the regret of not downloading the file when it is possible and waisting backhaul.
\begin{remark}
Note that this problem formulation is equivalent to the BMMG as the SBSs still take their decision independently. Moreover, in the BMMG, the decisions taken by one SBS do not depend on the identity of the SBS itself but on the mean number of requested files. Thus, introducing a set of virtual SBSs that take decisions independently on their real related SBSs, keeps the model valid. 
\end{remark}
Next, we study the proper mixed Nash equilibrium.\vspace{0.05cm}
\subsubsection{Proper-Mixed Strategies}
In this section we start by studying the formulated SBMMG and then extend the results for the BMMG. In particular, we are interested in finding the PMNE where all the SBSs take their decision probabilistically and do not have the incentive to deviate from their chosen strategy. In contrast to existing works on minority games \cite{challet2013minority} which provide conditions for the existence of mixed equilibria when $\phi=\frac{F-1}{2}$ with $F=2k+1$ and $k\in\mathbbm{N}_0$, here, we need to extend the results for any $\phi $ and $F\in \mathbbm{N}_0$. For the SBMMG, we have the following result.

To define the mixed strategies of the SBSs, we use the indifference principle  \cite{harsanyi1973games} that provides the condition that allows the players to reach a mixed strategy Nash equilibrium. For this, each player should be indifferent amongst each of the actions he puts non-zero weight on, yet he mixes them so as to make every other player is also indifferent. By using the indifference principle, a PMNE exists when the expected utility of requesting the predicted file is equal to the expected utility of not requesting that file, i.e., $u_n(c, \boldsymbol{p}_{-n})=u_n(d,\boldsymbol{p}_{-n})$. By substituting $u_n(d,f_d)$ based on its definition in (9) and equating the utilities (10) and (11), we get:
\begin{align}
\prod_{i\neq n}^{G}(1-p_i) u_n(c,1)+\sum_{i\neq n}^{G}p_i\prod_{j\neq \{i,n\}}^{G}(1-p_j) u_n(c,2)
+...+\prod_{l\neq n}^{N}p_lu_n(c,G)=0.
\end{align}


However, a PMNE may not be fair in the sense that some SBSs will request their file with higher probability compared to the other SBSs. In order to allow the SBSs to equally use the available backhaul, we are interested in looking for a fair PMNE.
\begin{proposition}
In the SBMMG, there exists a unique fair PMNE where all the SBSs select strategy $c$ with the same probability $p$.
\end{proposition}
\begin{proof}
When all the SBSs choose strategy $c$ with the same probability $p$, the utilities of selecting one of the strategies $c$ and $d$ write as follows:
\vspace{0.1cm}
\begin{align}
\bar{u}_n(c,\boldsymbol{p}_{-n})=\sum_{k=0}^{G-1}\dbinom{G-1}{k} p^{k}(1-p)^{G-(k+1)}u_n(c,k+1).\\
\bar{u}_n(d,\boldsymbol{p}_{-n})=\sum_{k=0}^{G-1}\dbinom{G-1}{k} p^{k}(1-p)^{G-(k+1)}u_n(d,G-k).
\end{align}

Using the indifference principle and (9), there exists a PMNE when:
\begin{equation}
\sum_{k=0}^{G-1}\dbinom{G-1}{k} p^{k}(1-p)^{G-(k+1)}u_n(c,k+1)=0.
\end{equation}
To prove that $p_i=p$, $\forall i\in\mathcal{N}$, is a solution for (15), assume the case where all the SBSs except SBS $n$ deviate from their mixed strategy and choose a pure strategy $c$ with probability $\sigma_i =0$, $\forall i\neq n$. Then, from (13) and (14), we have:
\begin{equation}
\bar{u}_n(c,\boldsymbol{\sigma}_{-n})=u_n(c,1)>u_n(d,G)=\bar{u}_n(d,\boldsymbol{\sigma}_{-n}).
\end{equation}
This comes from the fact that $u_n(d,G)=-u_n(c,1)$ and $u_n(c,1)>0$. Similarly, assume that all the SBSs except SBS $n$ select strategy $c$ with probability $\pi_i=1$, $\forall i\neq n$. From (13) and (14), we have:
\begin{equation}
\bar{u}_n(d,\boldsymbol{\pi}_{-n})=u_n(d,1)>u_n(c,G)=\bar{u}_n(c,\boldsymbol{\pi}_{-n}).
\end{equation}
From (16) and (17), we have:
\begin{equation}
\begin{cases}{}
\bar{u}_n(c,\boldsymbol{\sigma}_{-n})-\bar{u}_n(d,\boldsymbol{\sigma}_{-n})>0,\\
\bar{u}_n(c,\boldsymbol{\pi}_{-n})-\bar{u}_n(d,\boldsymbol{\pi}_{-n})<0.
\end{cases}
\end{equation}
By using the intermediate value theorem, we deduce that there exists a probability profile $\boldsymbol{p}=[p_1,...,p_G]$, with $p_i=p, ~~\forall i\in \mathcal{G}$, and $p\in(0,1)$ solving (12).
Since, the utility of selecting strategy $c$ in (13) is a decreasing function of $p$ and the utility of selecting strategy $d$ in (14) is an increasing function of $p$, then, the two utilities meet in only one point which is $p$.
\end{proof}

From Proposition 3, we can deduce the following:
\begin{corollary}
There exists a unique fair PMNE for the BMMG, where each SBS $n\in\mathcal{N}$, chooses a strategy profile $\boldsymbol{p}_n=[\mathcal{B}(1,F_n,p), \mathcal{B}(2,F_n,p),...,\mathcal{B}(i,F_n,p),...,\mathcal{B}(F_n,F_n,p)]$.
Here, $\mathcal{B}(i,F_n,p)$ is the binomial distribution and the probability $p$ is the same for all the SBSs in $ \mathcal{N}$.
\end{corollary}
\begin{proof}
Since at the fair PMNE in the SBMMG, each virtual/real SBS downloads one file with probability $p$, then the related real SBS $n$ in the BMMG, will decide to download each file from the $F_n$ files with an independent probability $p$. The probability of selecting $i$ files is hence given by the binomial distribution $\mathcal{B}(i,F_n ,p)$.
\end{proof}
The result in Corollary 1 does not follow directly from existing works on MGs \cite{challet2013minority} since those works are restricted to games with the same binary set of strategies for all the players, which is not the case in the BMMG. Moreover, the formulated BMMG has the independence characteristic, given in Remark 1, which is the main parameter that allows the introduction of the set of virtual SBSs and thus the derivation of Corollary 1. 

While analytically characterizing the uniqueness and properties of the PMNE is possible, the next step is to develop a practical algorithm that enables the SBSs to reach this PMNE or its neighborhood. To this end, we will next develop a reinforcement learning algorithm that converges to a refinement of the fair PMNE which is known as the Boltzmann-Gibbs equilibrium (BGE) \cite{lasaulce2011game}. The BGE is a special case of the $\epsilon$-Nash equilibrium which is a solution concept in which the players are within $\epsilon$ of the sought equilibrium. In other words, at an $\epsilon$-equilibrium no deviating SBS can improve its expected utility by a small amount $\epsilon$.

\section{Self-Organizing Learning Algorithm}
\label{algo}
To find an approximation of the PMNE, we propose an algorithm based on RL in which the players do not need to know any information about the actions of the other players. At each time period, the SBSs need to only observe an estimation of their utility and select their strategy accordingly. In contrast to works that use RL approaches such as \cite{bennis2013self}, in which the convergence to a BGE and the uniqueness are not ensured, in this work we prove that the RL algorithm converges to a unique BGE that approaches the PMNE of the formulated SBMMG. First, we need to define the notion of a smoothed best response.
\begin{definition}{} 
The smoothed best response function $\boldsymbol{\beta}_n^{(\kappa_n)}:[0,1]^{(G-1)\times 2}\to [0,1]^2$, with parameter $\kappa_n>0$, is defined as follows:
\begin{equation}
\boldsymbol{\beta}_n^{(\kappa_n)}(\boldsymbol{p}_{-n})=\Big(\beta_{n}^{(\kappa_n)}(c,\boldsymbol{p}_{-n}),\beta_{n}^{(\kappa_n)}(d,\boldsymbol{p}_{-n})\Big),
\end{equation}
and $\forall~~ a ~~\in \{c,d\}$, $\beta_{n}^{(\kappa_n)}(a,\boldsymbol{p}_{-n})$ is given by the Boltzmann-Gibbs distribution:
\begin{equation}
\beta_{n}^{(\kappa_n)}(a,\boldsymbol{p}_{-n})=\frac{\exp\Big(\kappa_{n}\bar{u}_n(a,\boldsymbol{p}_{-n})\Big)}{\exp\Big(\kappa_{n}\bar{u}_n(c,\boldsymbol{p}_{-n})\Big)+\exp\Big(\kappa_{n}\bar{u}_n(d,\boldsymbol{p}_{-n})\Big)}.
\end{equation}
\end{definition}

Here, we note that depending on the value of the parameter $\kappa_n$, the smoothed best response of SBS $n$ changes. In fact, as $\kappa_n\to0$, the smoothed best response of SBS $n$ converges to the uniform probability distribution, i.e., $\boldsymbol{\beta}_n^{(\kappa_n)}(\boldsymbol{p}_{-n})=(1/2,1/2)$, irrespective the strategies adopted by all the other players. However, when $\kappa_n\to\infty$, the smoothed best response is a uniform probability distribution over the pure strategies that are best responses to the strategies adopted by all the other players. The parameter $\kappa_n$ represents the exploitation/exploration rate that enables the small base station to make a decision about whether to just exploit by always selecting the action with the maximum utility in the current stage, or fold in some exploration and try other actions to discover more information about the network that can be used to achieve better long-term rewards. During the exploitation period, the SBSs do not stop the learning process as they will still use the reward received at the reached sub-slot, to adapt its behavior in the future sub-slots. However, the SBSs may be blocked at local minima which prevent them from reaching the highest possible utility. At the exploration phase, the SBSs determines which action to choose so that the SBSs learn perfectly the best actions that will allow them to determine how behave in the future. Eventually, when everything to know is learned by th SBSs, there is no need to continue the exploration, and the SBSs must act optimally according to the best learned and possible policy. The desired algorithm must allow the exploration probability to decrease as the SBSs gather enough information and the network is better known. This can enable for learning the optimal policy by the end of the time period. The value of $\kappa_n$ is commonly chosen to be $\frac{1}{t}$, where $t$ is the current sub-slot.

Now, we define the Boltzmann-Gibbs equilibrium (BGE) which is also known as the logit equilibrium as follows:
\begin{definition}
The strategy profile $\boldsymbol{p}^*=[\boldsymbol{p}_1^*,...,\boldsymbol{p}_G^*]$ is a \emph{BGE} with parameter $\kappa_k>0$ of the game if $\forall ~~n\in \mathcal{G}$,
\begin{equation}
\boldsymbol{p}_n^*=\boldsymbol{\beta}_n^{(\kappa_n)}(\boldsymbol{p}^*_{-n}).
\end{equation}
\end{definition}
In an SCN, the decentralized nature of the operations at the SBSs makes it difficult for each SBS to communicate with all the other SBSs and get access to their selected strategies. Thus, the SBSs cannot compute the exact value of their utility at a given time period. Hence, we assume that the utilities of the SBSs are subject to random error and an SBS $n$ can only observe an estimation $\widetilde{u}_n(a_n(t))$ of its utility function $u_n(a_n(t),\boldsymbol{a}_{-n}(t))$ for selecting action $a$ at time $t$ when all the other SBSs select the actions given in $\boldsymbol{a}_{-n}$. The estimated utility is given as follows: 
\begin{equation}
\widetilde{u}_n(a_n(t))=u_n(a_n(t),\boldsymbol{a}_{-n}(t))+\epsilon_{n,a_n(t)}(t),
\end{equation}
where $\forall n \in \mathcal{G}, ~~a_n(t)\in \{c,d\}$, and $\epsilon_{n,a_n(t)}$ is a random variable that represents the distribution of the observation error on the instantaneous utility when SBS $n$ selects action $a$. Its expected value is assumed to be 0, i.e., $\mathbbm{E}\Big[\epsilon_{n,a_n(t)}\Big]=0$.  The variable $t$ is used to denote the number of the current sub-slot as a given  time period is divided into multiple sub-slots  based on the classical wireless frame structure. It should be noted that a time period is the required time for the SBS to receive via the backhaul all the files it decides to download. To estimate the instantaneous utility, an SBS does not require any communication with the other SBSs nor with a centralized entity. The value of $\widetilde{u}_n(a_n(t))$ can be computed based on the received power from the serving macro base station as well as the perceived  interference power from all the other macro base stations when serving other SBSs over the same backhaul resource blocks.

Knowing the estimated value of the utility at each time period as well as the selected action, each SBS can estimate the  achievable expected utility for each of the actions in its strategy set. Based on these estimations, the SBS selects the action with the highest expected utility. As a consequence of the Boltzmann-Gibbs distribution, by adapting the parameters $\kappa$, the SBSs can be pushed to explore new actions and not always select the ones achieving the highest performance. This allows the SBSs to try all their set of actions looking for any possible improvement of the expected utility. In what follows, we present a decentralized RL algorithm for backhaul management in which the SBSs simultaneously learn both utilities and strategies. First, we denote the estimate of the expected utility of an SBS $n$ by:
 \begin{equation}
\hat{ \boldsymbol{u}}_n(t)=\Big( \hat{u}_n(c,t), \hat{u}_n(d,t)\Big),
 \end{equation}
 where $\hat{u}_n(a,t)$ is the estimate of $\hat{u}_n(a,\boldsymbol{p}_{-n}(t))$. The expected utility is updated at each time period based on the instantaneous observations $\boldsymbol{\widetilde{u}}_n(t)$ while the probability of selecting each action is a function of the estimated smoothed best response function. Before providing the RL algorithm, we define the estimate best response function $\boldsymbol{\widetilde{\beta}}_n^{\kappa_n}: \mathbbm{R}^2\to \mathbbm{R}_+$ based on the estimated utility function $\boldsymbol{\hat{u}}_n$ as follows:
\begin{equation}
\boldsymbol{\widetilde{\beta}}_n^{\kappa_n}(\boldsymbol{\hat{u}}_n(t))=\Big(\widetilde{\beta}_n^{\kappa_n}(c,\boldsymbol{\hat{u}}_n(t)),\widetilde{\beta}_n^{\kappa_n}(d,\boldsymbol{\hat{u}}_n(t))\Big),
\end{equation}
with $\boldsymbol{\hat{u}}_n(t)=[\hat{u}_n(c,t),\hat{u}_n(d,t)]$ and
\begin{equation}
\widetilde{\beta}_n^{\kappa_n}(a,\boldsymbol{\hat{u}}_n(t))=\frac{\exp(\kappa_n \hat{u}_n(a,t))}{\exp(\kappa_n \hat{u}_n(c,t))+\exp(\kappa_n \hat{u}_n(d,t))},
\end{equation}
where $a\in\{c,d\}$.

As first proposed in \cite{leslie2003convergent}, any RL algorithm can be defined as follows ($\forall n\in\mathcal{G}$, $\forall a\in \{c,d\})$:

\begin{equation}
\left\{
\begin{array}{l}
  \hat{u}_n(a,t)=\hat{u}(a,t-1)+\\
  +\alpha (t)\mathbbm{1}_{\{a_n(t)=a\}}\Big(\widetilde{u}_n(a(t))-\hat{u}_n(a,t-1)\Big),\\
p_n(t)= p_n(t-1)+\lambda_n(t)\Big(\widetilde{\beta}^{(\kappa_n)}_n(\boldsymbol{\hat{u}}_n(t))-p_n(t-1)\Big),
\end{array}
\right.
\end{equation}
where, $(\boldsymbol{\hat{u}}_n(0), \boldsymbol{p}_n(0))\in \mathbbm{R}^{2}\times [0,1]^{2}$, is an arbitrary initialization of player $n$. For instance, $\boldsymbol{\hat{u}}_n(0)=(0,0)$ and $\boldsymbol{p}_n=(1/2,1/2)$, can be defined as the initial values. Moreover, the following conditions must be satisfied for all $(j,n) \in \mathcal{G}^2$:
\begin{numcases}{}
 \lim_{T\to\infty}{\sum_{t=1}^{T}\alpha_n(t)}= +\infty, \lim_{T\to\infty}{\sum_{t=1}^{T}\alpha^2_n(t)}<+\infty,\nonumber\\
 \lim_{T\to\infty}{\sum_{t=1}^{T}{\lambda_n(t)}=+\infty},  \lim_{T\to\infty}{\sum_{t=1}^{T}{\lambda^2_n(t)}<+\infty},\\
 \lim_{t\to\infty}{\frac{\lambda_n(t)}{\alpha_n(t)}=0}.\nonumber
\end{numcases}
and either,
\begin{numcases}{}
 \forall n\in\mathcal{G}, ~~~~\lambda_n(t)=\lambda(t),~~\nonumber \text{ or},\\
\forall n \in\mathcal{G}\setminus\{G\}, ~~~~\lim_{t\to\infty}{\frac{\lambda_n(t)}{\lambda_{n+1}(t)}=0}.
\end{numcases}
It should be noted that the proposed RL algorithm is (26) does not require any assumption on the dynamics of the channel state as it is able to capture the changes of both channel and users' demands. In fact, in contrast to other learning algorithm, the SBSs in the proposed algorithm decide whether to download predicted files or not based on the estimation of their utility that is computed using the channel and traffic statistics from a number of previous sub-slots with different channel state and demand profiles. For RL algorithms that follow (26), there is no guarantee of convergence to an equilibrium even if the algorithm has a steady point \cite{rose2011learning,leslie2003convergent}. Moreover, most works that are able to prove convergence to an equilibrium such as \cite{bennis2013self}, cannot guarantee the uniqueness of this equilibrium. For the studied SBMMG, we provide the following result on the convergence of the algorithm in (26) to a unique BGE.
\begin{theorem}
The algorithm in (26) converges to a unique BGE with parameter $\kappa_n,~~ \forall n\in\mathcal{G}$, in the SBMMG and we have:
\begin{equation}
\left\{
\begin{array}{l}
\lim_{t\to\infty}\boldsymbol{p}_n(t) =\boldsymbol{p}_n^*,\\
\lim_{t\to\infty}\hat{u}_n(a,t)=\bar{u}_n(a,\boldsymbol{p}_{-n}^*).
\end{array}
\right.
\end{equation}
\end{theorem}
\begin{proof}
The proof is given in the Appendix.
\end{proof}
The main challenge is to prove that the proposed algorithm converges to a unique point for the formulated game and show that this point corresponds to the BGE. In contrast to most of work that use similar RL algorithms, we do not only prove that the algorithm converges but we also ensure that the reached equilibrium is unique for the formulated game. To this end, we first write (26) as an approximated ordinal differential equation and prove that it admits at least one rest point. Then, to show the uniqueness of the rest point, we prove that the defined ordinal differential equation is a contraction which in turn guarantees the uniqueness of equilibrium.

Since a BGE is a special case of $\epsilon$-Nash equilibrium, the BGE is an approximate equilibrium of the fair PMNE that is within $\epsilon$ of the PMNE. Here, we provide a bound for the utility improvement an SBS can obtain by unilaterally deviating from the BGE.
\begin{proposition}
At the BGE, assume the strategy profile of the SBMMG with parameters $\kappa_n>0$ is $\boldsymbol{p}^*$. Then, 
$\boldsymbol{p}^*$ is an $\epsilon$-equilibrium with $\epsilon=\frac{1}{\kappa_n}\text{ln}(1)$.
\end{proposition}
\begin{proof}
This result follows directly from \cite{young2004strategic} based on the definition of the best response $\boldsymbol{\beta}_n^{(\kappa_n)}(\boldsymbol{p}^*_{-n})$ which writes as follows:
\begin{equation*}
\begin{array}{l}
\boldsymbol{\beta}_n^{(\kappa_n)}(\boldsymbol{p}^*_{-n})=\text{arg} \max_{\boldsymbol{p}_{-n}\in[0,1]^2}\Big[\bar{u}_n(\boldsymbol{p}_n,\boldsymbol{p}_{-n}^*)-H(\boldsymbol{\beta}_n^{(\kappa_n)}(\boldsymbol{p}^*_{-n}))\Big],
\end{array}
\end{equation*}
where $H(\boldsymbol{\beta}_n^{(\kappa_n)}(\boldsymbol{p}^*_{-n}))=\frac{1}{\kappa_n}\Big(p_n^*\text{log}(p_n^*)+(1-p_n^*)\text{log}(1-p_n^*)\Big)$, and the following property of the entropy function $H(\boldsymbol{\beta}_n^{(\kappa_n)}(\boldsymbol{p}^*_{-n}))$: 
\begin{equation*}
H(\boldsymbol{\beta}_n^{(\kappa_n)}(\boldsymbol{p}^*_{-n}))-H(\boldsymbol{p}_{-n})\leq H(\boldsymbol{p}_n^0),
\end{equation*}
where, $\boldsymbol{p}_n^0=(\frac{1}{2},\frac{1}{2})$ is the initial uniform distribution over the set of strategies. Thus, we can deduce the following:
\begin{equation*}
\begin{aligned}
\bar{u}_n(\boldsymbol{p_k},\boldsymbol{p^*_{-k}})-\bar{u}_n(\boldsymbol{\beta}_n^{(\kappa_n)}(\boldsymbol{p}^*_{-n}),\boldsymbol{p^*_{-k}})
\leq \frac{1}{\kappa_n}\Big( H(\boldsymbol{\beta}_n^{(\kappa_n)}(\boldsymbol{p}^*_{-n}))-H(\boldsymbol{p}_{n})\Big)&\leq -\frac{1}{\kappa_n}\text{log}(\frac{1}{2}),\\
&\leq \frac{1}{\kappa_k}\text{log}(2).
\end{aligned}
\end{equation*}
\end{proof}
The BGE is equal to the fair PMNE when $\epsilon=0$. Thus, by putting $\kappa_n \to \infty$ while satisfying the condition given in the Appendix, we can make sure to approach the fair PMNE.
\section{Simulation Results and Analysis}
\label{sim}
For our simulations, we consider a $2$~km $\times\ 2 $~km area which is covered by two MBSs and five SBSs. The SBSs are connected to the MBSs via a heterogeneous backhaul having a total capacity of $1$ Gbps unless stated otherwise. In order to ensure fairness between the SBSs, we always set the minimum number of backhaul resource blocks equal to the number of SBSs. Then, without loss of generality, we use a matching algorithm similar to the one proposed in \cite{semiari2015matching}, for allocating backhaul capacity. The total number of predicted files is set to 150 randomly distributed over the SBSs. All statistical results are averaged over 100 independent runs. We assume that the channel is static over a large number of sub-slots knowing that the expected utility is only averaged over the possible strategies. Based on the required number of iterations for the convergence of the algorithm in Fig \ref{fig:conv1}, we can deduce that such an assumption is reasonable for our model as the channel is known to be static for 10000 iterations in practical 4G systems. However, to avoid the overhead due to the frequent cache updates, the algorithm can be run over both strategies and channel statistics. For instance, to operate over slower time scales, the utility function can be redefined as the expectation over all the channel realizations over a given time duration $\mathbbm{E}_{t}{[\tilde{u}(a_n(t),h(t))}]$.

Once the utilities defined based on the backhaul allocation algorithm, we use the proposed RL algorithm, abbreviated by BMRL, with $\alpha_n(t) = \frac{1}{t}$ and $\lambda_n(t)=\frac{1}{t^2}$, to reach the BGE. The algorithm is run until convergence for different configurations. To show the performance gain of the proposed framework, we compare the decentralized BMRL with a centralized greedy algorithm (CGA). In CGA, a central entity receives information from all the SBSs regarding the number of their current and predicted requests. It is also aware of the capacity of the backhaul links. At each iteration, the central entity allows some SBSs to download a fixed number of predicted files. The chosen SBSs at each iteration are determined based on a fairness rule, i.e., the SBSs that have the lowest number of downloaded predicted files are selected. We also compare the BMRL with an ideal and optimal centralized algorithm (OCA) which is similar to CGA in which there is no information exchange between the SBSs and the central entity. The CGA is optimal as it never exceeds the capacity of the backhaul and guarantees fairness between the SBSs. Even though, CGA is not realistic or practical, it allows benchmarking the proposed approach against an ideal and optimal scheme. 

In our simulations, we use the backhaul allocation algorithm proposed in [3] for all scenarios and the comparison between the different approaches is with respect to the impact of downloading predicted files on the served urgent requests by the SBSs when using the three different algorithms. In this sense, the ``optimal'' centralized algorithm is considered as optimal because it has complete knowledge of the actions of the SBSs and their local information without accounting for the impact of the used algorithm in [3]. Having such information, the centralized entity is able to determine the global optimal solution given the outcome of the algorithm in [3].

In Fig.~\ref{fig:ut1}, we assess the impact of the parameter $\kappa_n$ on the achievable utility by the SBSs in the BMRL. Fig.~\ref{fig:ut1} shows the the variation of the difference between the available backhaul capacity and the required capacity for serving all the predicted files that are requested by the SBSs, while increasing the number of predicted files in the network. Note that the available backhaul capacity is the same for all the configurations while the number of SBSs having predicted files is being increased.
\begin{figure}
\centering
\includegraphics[scale=0.5]{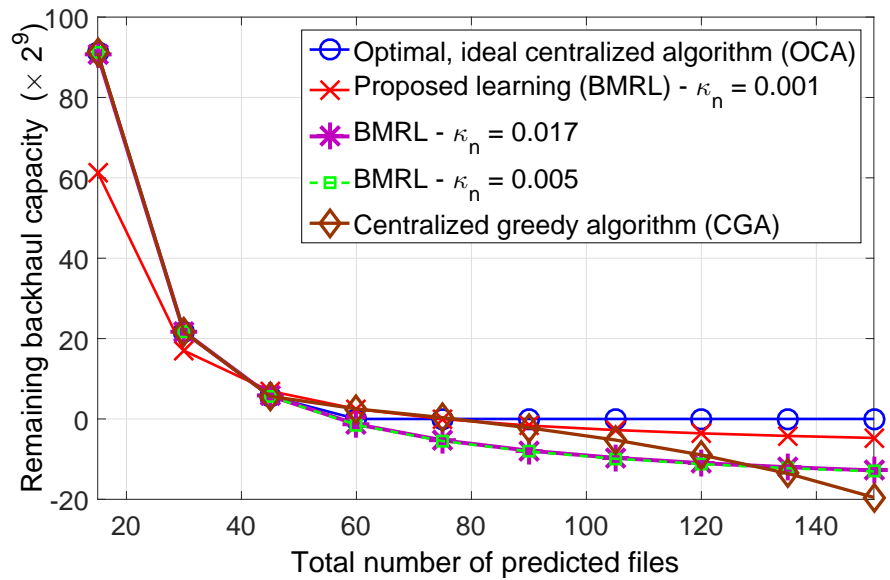}
\caption{Difference between the available backhaul capacity and the required capacity for serving the requested files with respect to the number of files and $\kappa_n$.}
\label{fig:ut1}
\end{figure}
By increasing the number of files in the network, the allocated backhaul to each SBS decreases resulting in a decreasing average utility. In Fig.~\ref{fig:ut1}, we can see that the parameter $\kappa_n$ has a significant impact on the performance of the proposed algorithm. In fact, when the value of $\kappa_n$ is high, the SBSs tend to choose more frequently the actions that are best responses to the actions of the other SBSs. Hence, when the backhaul capacity allows to serve all the predicted files, i.e., the total number of files is up to 60 in Fig.~\ref{fig:ut1}, higher values of $\kappa_n$ ($\kappa_n=0.005$ and $\kappa_n=0.017$) allow the BMRL to achieve the same performance as the OCA. In this case, BMRL is also as good as CGA due to the available backhaul that can support the extra exchanged packets in CGA. In contrast, decreasing the value of $\kappa_n$ will lead the SBSs to play all the actions uniformly. In this case, the performance of the BMRL is close to OCA and much higher than CGA when the predicted files cannot all be served through the backhaul. In fact, by increasing the number of files, more backhaul is allocated for the information exchange resulting in a decreasing performance in CGA. Thus, by properly choosing the values of the parameter $\kappa_n$, we can achieve optimal performance at a lower signaling overhead compared to CGA. 

In Fig.~\ref{fig:ut2}, the performance of the BMRL is compared with OCA in three different cases:
\begin{itemize}
\item \emph{Case 1:} The available backhaul capacity is higher than the required capacity to serve the current requests, but the extra backhaul capacity can only be used to serve up to 60 predicted files.
\item \emph{ Case 2:} The backhaul capacity (50 Mbps) is lower than the required capacity to serve the current requests.
\item \emph{ Case 3:} The backhaul capacity (3 Gbps) is sufficient to serve all the current requests and up to 150 predicted files.
\end{itemize}

For cases 2 and 3, choosing relatively high values of $\kappa_n$ allows one to achieve exactly the same performance as the OCA. In Case 2, by choosing $\kappa_n=0.001$, the SBSs download all their predicted files with a high probability approaching $p_n=1$. On the other hand, in case 3, by putting $\kappa_n=1$, none of the SBSs requests a predicted file and the probability of requesting a predicted file approaches $p_n=0$. Finally, in case 1, when the capacity of the backhaul is not sufficient for serving all the predicted files, $\kappa_n$ should be chosen carefully depending on the backhaul capacity and the approximate total number of files in the network. For this case, $\kappa_n=0.001$ according to Fig.~$\ref{fig:ut1}$.

\begin{figure}
\centering
\includegraphics[scale=0.5]{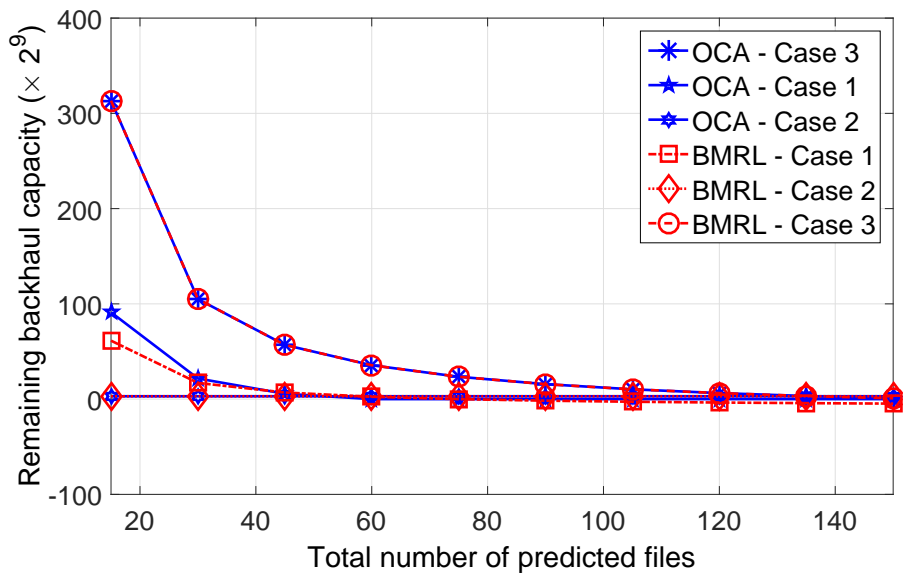}\vspace{-0.3cm}
\caption{The difference between the available backhaul capacity and the required capacity for serving the predicted files that are requested by the SBSs with respect to the number of files.}\vspace{-0.6cm}
\label{fig:ut2}
\end{figure}

Fig.~\ref{fig:file} shows the total amount of predicted data that is requested by the SBSs in BMRL, OCA, CGA and a centralized random fair algorithm (RFA) in which neither the capacity of the backhaul nor the requests profile of the SBSs are known to the central entity. We can observe that the backhaul usage in BMRL is similar to the backhaul usage in OCA in 85 \% of the cases while it outperforms CGA and RFA by up to 50\% in terms of the amount of cached content and the rate with which the current requests are served, respectively. In fact, when the available backhaul capacity is higher than the total number of predicted files in the network (up to 60 predicted files), in CGA, BMRL and OCA all the predicted files are requested. The backhaul usage in the RFA is lower compared to the other algorithms, because the capacity of the backhaul is selected randomly at each iteration and since the backhaul is allocated fairly to the SBSs, each SBS downloads the predicted file with probability $p_n=0.5$. This results in an inefficient backhaul usage whether the backhaul is available or not. When the capacity of the backhaul is not sufficient to support all the predicted files, the amount of downloaded content in CGA decreases by increasing the number of files in the network which is due to the extra packets that are transmitted over the backhaul for coordination with the central entity. In this same case, both OCA and BMRL allow the SBSs to download files without exceeding the capacity of the backhaul. Note that in Fig.~\ref{fig:file}, the values of $\kappa_n$ were chosen based on the maximum amount of files that can be downloaded without exceeding the capacity of the backhaul.
\begin{figure}
\centering
\includegraphics[scale=0.5]{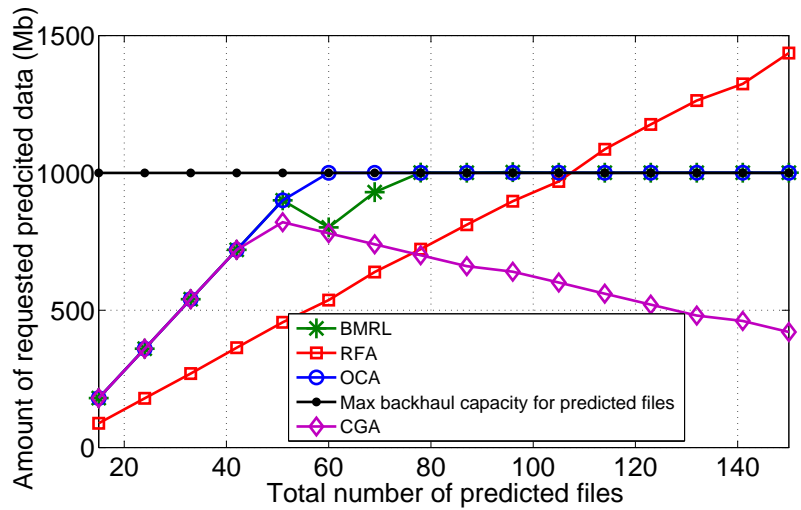}
\caption{Amount of requested predicted data with respect to the backhaul capacity.}
\label{fig:file}
\end{figure}

Fig.~\ref{fig:conv1} presents the number of iterations needed for convergence to the BGE. Fig.~\ref{fig:conv1} shows that the BMRL requires only 2 iterations to converge when no backhaul is available for serving the predicted requests (Case 2). Around $532$ iterations, are needed when there is an available backhaul for downloading all the predicted files. Fig. ~\ref{fig:conv1} shows that, for Case 3, the number of iterations decreases with the number of predicted files until reaching 111 iterations. In Case 1, the number of iterations decreases in the beginning since all the files can be served (similarly to Case 3) and the number of iterations increases again when it is not possible to serve all the files due to the SBSs' probabilities of using the backhaul that variate largely before reaching the equilibrium. The number of iterations can be considered as quite reasonable compared to other works that use RL approaches for SCNs such as \cite{bennis2013self}, in which the number of iterations that is required for convergence exceeds 5000. Moreover, given that the duration of a single time slot is around 1 ms and the mean required time for one iteration is 100 ns, the number of iterations that is required for the convergence of the algorithm is acceptable for the game we consider.

\begin{figure}
\centering
\includegraphics[scale=0.55]{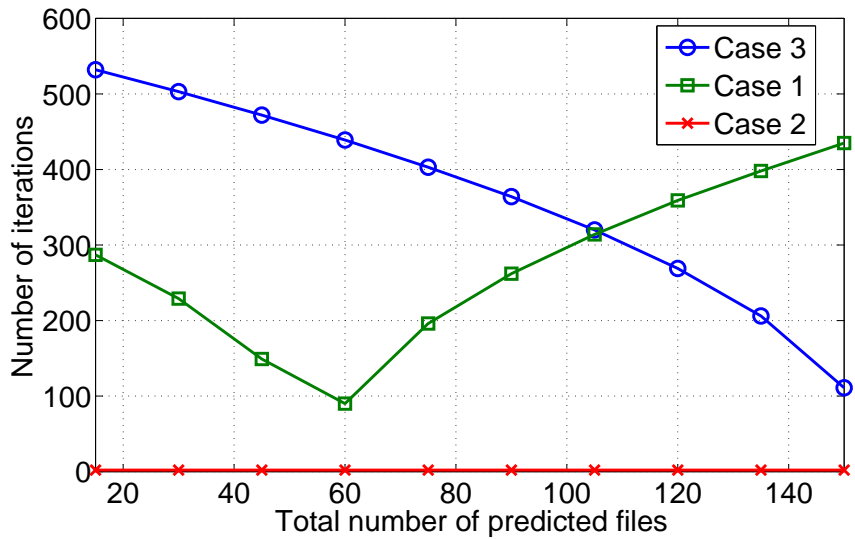}
\caption{Number of iterations for convergence to the logit equilibrium.}
\label{fig:conv1}
\end{figure}

\section{Conclusion}
\label{conclusion}
In this paper, we have proposed a novel backhaul management approach for cache-enabled small cell networks, while taking into account the heterogeneity of the backhaul links. We divided the requests of each SBS into predicted requested that can be served locally by the SBSs and current requests that must be served instantaneously from the backhaul. We have formulated a minority game in which the SBSs have to define the number of predicted files to download at each time period without deteriorating the QoS of the current requests. We have proved the existence of multiple pure Nash equilibria and the existence of a unique fair mixed Nash equilibrium allowing all the SBSs to use the backhaul evenly. Moreover, we have proposed a self-organizing reinforcement learning algorithm that reaches a unique Boltzmann-Gibbs equilibrium that approximates the PMNE. Simulation results have shown that the proposed algorithm outperforms the centralized greedy algorithm and its performance is exactly the same as the performance of the ideal and optimal algorithm in more than 85 \% of the cases. The impact of the caching and replacement policies at the SBSs, on the number of predicted files as well as the whole backhaul management approach, is left for future work.

\bibliographystyle{IEEEtran}
\bibliography{references}

\begin{appendix}

\section{Proof of Theorem 1}
The proof of the convergence of the algorithm in (26) to a unique BGE can be decomposed into two parts. We first start by proving that the algorithm converges surely to a BGE and then prove that the algorithm admits a unique steady point which corresponds to the unique BGE. 

Since a minority game is a special case of congestion games and based on \cite[Theorem 3.1]{monderer1996potential}, we can deduce that the formulated SBMMG admits a potential function. On the other hand, based on \cite[Theorem 7]{leslie2003convergent} and knowing that the SBMMG accepts a finite number of BGEs, then the algorithm in (26) admits at least one steady point and converges with probability 1 to a BGE.

In order to prove the uniqueness of the steady point of the algorithm in (26) we analyze the Robin-Monro iteration form of (26) \cite{benaim1999dynamics}. The limiting ordinal differential equations (ODE) of the Robin-Monro equations write as follows:
\begin{equation}
\left\{
\begin{array}{l}
  \dot{\hat{u}}_n(a,t)=\mathbbm{E}_{\boldsymbol{p}}\Big[\widetilde{u}_n(a(t))\Big]-\hat{u}_n(a,t),\\
\dot{p}_n(t)= \widetilde{\beta}^{(\kappa_n)}_n(\boldsymbol{\hat{u}}_n(t))-p_n(t-1),
\end{array}
\right.
\end{equation}
Given the existence of at least one fixed point for the ODE function:
\begin{equation}
\dot{\boldsymbol{p}}_n=\boldsymbol{\beta}_n^{(\kappa_n)}(\boldsymbol{p}_{-n})-\boldsymbol{p}_n,
\end{equation}
then, we have:
\begin{equation}
\boldsymbol{p}_n^*=\boldsymbol{\beta}_n^{(\kappa_n)}(\boldsymbol{p}_{n}^*),
\end{equation}
and by replacing with (22) and (32) in $\dot{\hat{u}}_n(a,t)$, we get:
\begin{equation}
  \dot{\hat{u}}_n(a,t)=\mathbbm{E}_{\boldsymbol{p}^*}\Big[u_n(a(t),\boldsymbol{p}_{-n}^*)\Big]+\mathbbm{E}_{\boldsymbol{p}^*}\Big[\epsilon_{n,a_n(t)}(t)\Big]-\hat{u}_n(a,t),
\end{equation}
which reduces to solving the ODE:
\begin{equation}
  \dot{\hat{u}}_n(a,t)=u_n(a(t),\boldsymbol{p}_{-n}^*)-\hat{u}_n(a,t),
\end{equation}
Now, we prove the existence of a unique fixed point for the ODE (34). Given Banach fixed point theorem which says that a contraction has a unique fixed point, it is sufficient to prove that the ODE in (34) is a contraction in order to prove the uniqueness of the fixed point $u_n(a(t),\boldsymbol{p}_{-n}^*)$ given by:
\begin{equation} 
\bar{u}_n(c,\boldsymbol{p}_{-n})=\sum_{k=0}^{G-1}\dbinom{G-1}{k} p^{k}(1-p)^{G-(k+1)}u_n(c,k+1),
\end{equation}
with $\boldsymbol{p}_n=\boldsymbol{\beta}_n^{(\kappa_n)}(\boldsymbol{p}_{-n}).$
\begin{definition}{(Contraction):}
A map function $g:X\to X$ is said to be a $\theta$-contraction if $\exists~~ 0<\theta<1$ such that:
\begin{equation}
|g(x_1,x_2)|\leq \theta |x_1-x_2|.
\end{equation}
\end{definition}
Consider the difference $|\bar{u}_n(c,\boldsymbol{p}_{-n})-\bar{u}_n(c,\boldsymbol{p}^{\prime}_{-n})|$, we have:
\begin{equation}
\begin{array}{l}
|\bar{u}_n(c,\boldsymbol{p}_{-n})-\bar{u}_n(c,\boldsymbol{p}^{\prime}_{-n})|=\\
 |p_n\sum_{k=0}^{G-1}\dbinom{G-1}{k} p^{k}(1-p)^{G-(k+1)}u_n(c,k+1)-\\
 p^{\prime}_n\sum_{k=0}^{G-1}\dbinom{G-1}{k} (p^{\prime})^{k}(1-p^{\prime})^{G-(k+1)}u_n(c,k+1)|=\\
\Big |\sum_{k=0}^{G-1}u_n(c,k+1)(p_n-p^{\prime}_n)\dbinom{G-1}{k} (p^{\prime})^{k}(1-p^{\prime})^{G-(k+1)}\Big|\\
 \leq |(p_n-p^{\prime}_n)\sum_{k=0}^{G-1}u_n(c,k+1)|,\\
 \leq |\sum_{k=0}^{G-1}u_n(c,k+1)||p_n-p^{\prime}_n|.
\end{array}
\end{equation}
By replacing with the best response functions we get:
\begin{equation}
\begin{array}{l}
|p_n-p^{\prime}_n|=|\beta_{n}^{(\kappa_n)}(c,\boldsymbol{p}_{-n})-\beta_{n}^{(\kappa_n)}(c,\boldsymbol{p}_{-n})|\\
=\Big|\frac{\exp\Big(\kappa_{n}u_n(a,\boldsymbol{p}_{-n})\Big)}{\exp\Big(\kappa_{n}u_n(c,\boldsymbol{p}_{-n})\Big)+\exp\Big(\kappa_{n}u_n(d,\boldsymbol{p}_{-n})\Big)}
-\frac{\exp\Big(\kappa_{n}u_n(a,\boldsymbol{p}_{-n})\Big)}{\exp\Big(\kappa_{n}u_n(c,\boldsymbol{p}_{-n})\Big)+\exp\Big(\kappa_{n}u_n(d,\boldsymbol{p}_{-n})\Big)}\Big|
\end{array}
\end{equation}
After some numerical computation we get:
\begin{equation}
\begin{array}{l}
|p_n-p^{\prime}_n|\leq \kappa_n||\boldsymbol{p}-\boldsymbol{p}^{\prime}||_{\infty}
\end{array}
\end{equation}
by replacing in (37) we have:
\begin{equation*}
|\bar{u}_n(c,\boldsymbol{p}_{-n})-\bar{u}_n(c,\boldsymbol{p}^{\prime}_{-n})|\leq \kappa_n|\sum_{k=0}^{G-1}u_n(c,k+1)|||\boldsymbol{p}-\boldsymbol{p}^{\prime}||_{\infty}.
\end{equation*}
Since we have $u_n(c,\boldsymbol{p}_{-n})u_n(d,\boldsymbol{p}_{-n})\leq0$, we can conclude that $u_n(\boldsymbol{p})$ is an $\infty$-contraction and admits a unique fixed point if $\kappa_n\leq |\sum_{k=0}^{G-1}u_n(c,k+1)|$.
Following the results from stochastic approximation algorithms and considering the Lyapunov function $V(\boldsymbol{p})=||\boldsymbol{p}-\boldsymbol{p}^*||_{\infty}$ for ODE (31), we deduce that $\boldsymbol{p}^*$ is the unique globally asymptotically stable point of (31). Thus, the formulated SBMMG admits a unique fixed point which is the BGE of the game. 
\end{appendix}

\end{document}